\documentclass[11pt,preprint]{article}

\usepackage{microtype}
\usepackage[T1]{fontenc}
\usepackage[utf8]{inputenc}
\usepackage{fullpage}
\usepackage[margin=1in]{geometry}
\usepackage{amsthm,amssymb,amsmath}  
\usepackage{xspace,enumerate}
\usepackage[dvipsnames]{xcolor}
\usepackage[colorlinks=true,urlcolor=Blue,citecolor=Green,linkcolor=BrickRed]{hyperref}
\usepackage[utf8]{inputenc}
\usepackage{thmtools}
\usepackage{thm-restate}
\usepackage{authblk}
\usepackage{todonotes}
\usepackage[noadjust]{cite}
\usepackage{subcaption}
\usepackage{framed}
\usepackage{booktabs}
\usepackage[noline,noend,linesnumbered,ruled]{algorithm2e} 
 
\theoremstyle{plain}
\newtheorem{theorem}{Theorem}
\newtheorem{lemma}[theorem]{Lemma}

\title{Incomplete Directed Perfect Phylogeny in Linear Time}
\author[1]{Giulia Bernardini}
\author[1]{Paola Bonizzoni}
\author[2]{Paweł Gawrychowski}
\affil[1]{DISCo, Universit\`{a} degli Studi Milano - Bicocca, Italy}
\affil[2]{Institute of Computer Science, University of Wrocław, Poland}

\newcommand{\cT}{\mathcal{T}}
\newcommand{\cO}{\mathcal{O}}

\newcommand{\cA}{\mathcal{A}}
\newcommand{\cB}{\mathcal{B}}

\newcommand{\DC}{\textsc{DC}}
\newcommand{\DT}{\textsc{DT}}
\newcommand{\IDP}{\textsc{IDPP}}

\newcommand{\comp}{\textnormal{\textsf{components}}}

\newcommand{\Dynproblem}[4]{
\begin{framed}
  \noindent
  \textbf{Problem:} \textsc{#1}

  \noindent
  \textbf{Input:} #2

  \noindent
  \textbf{Update:} #3
  
  \noindent
  \textbf{Query:} #4  
\end{framed}
}

\SetKwInput{Algorithm}{Algorithm}
\SetEndCharOfAlgoLine{}
\begin{document}
\date{}
\maketitle

\thispagestyle{empty}

\begin{abstract}
Reconstructing the evolutionary history of a set of species is a central task in computational biology. In real data, it is often the case that some information is missing: the Incomplete Directed Perfect Phylogeny (\IDP) problem asks, given a collection of species described by a set of binary characters with some unknown states, to complete the missing states in such a way that the result can be explained with a perfect directed phylogeny.
Pe'er et al. proposed a solution that takes $\tilde{\cO}(nm)$ time for $n$ species and $m$ characters. Their algorithm relies on pre-existing dynamic connectivity data structures: a computational study recently conducted by Fern{\'a}ndez-Baca and Liu showed that, in this context, complex data structures perform worse than simpler ones with worse asymptotic bounds.

This gives us the motivation to look into the particular properties of the dynamic connectivity problem in this setting, so as to avoid the use of sophisticated data structures as a blackbox.
Not only are we successful in doing so, and give a much simpler $\tilde{\cO}(nm)$-time algorithm for the \IDP{} problem; our insights into the specific structure of the problem lead to an asymptotically faster algorithm, that runs in optimal $\cO(nm)$ time. 
\end{abstract}

\clearpage
\setcounter{page}{1}

\section{Introduction}
A rooted phylogenetic tree models the evolutionary history of a set of species: the leaves are in a one-to-one correspondence with the species, all of which have a common ancestor represented by the root. 
A way of describing the species is by a set of characters that can assume several possible states, so that each species is described by the states of its characters. Such a representation is naturally encoded by a matrix $\cA$, $a_{i,j}$ being the state of character $j$ in species $i$.

When, for each possible character state, the set of all nodes that have the same state induces a connected
subtree, a phylogeny is called \emph{perfect}. The problem of reconstructing a perfect phylogeny from a set of species is known to be linearly-solvable in the case when the characters are binary~\cite{gusfield1991efficient}, and it is NP-hard in the general case~\cite{bodlaender2000hardness}. A popular variant of binary perfect phylogeny requires that the characters are directed, that is, on any path from the root to a leaf a character can change its state from $0$ to $1$, but the opposite cannot happen~\cite{camin1965method}.

In this paper, we study the Incomplete Directed Perfect Phylogeny problem (\IDP{} for short) introduced by Pe'er et al.~\cite{pe2004incomplete}, assuming that the characters are binary, directed, and can be gained only once.
The input of this problem is a matrix of character vectors in which some character states are unknown, and the question is whether it is possible to complete the missing states in such a way that the result can be explained with a directed perfect phylogeny.

\paragraph{Related work.} Besides being relevant in its own right~\cite{nikaido1999phylogenetic, bashir2005orthologous,6819844, satya2008,stevens2010reducing}, the problem of handling phylogenies with missing data arises in 
various tasks of computational biology, like resolving genotypes with some missing information into haplotypes~\cite{kimmel2005incomplete} and inferring tumor phylogenies from single-cell sequencing data with mutation losses~\cite{satas2020scarlet}. 
A generalization of the perfect phylogeny model where a character can be gained only once and can be lost at most $k$ times, called the $k$-Dollo  model~\cite{bonizzoni2012binary,gusfield2015persistent, bonizzoni2017beyond,el2018sphyr}, has also been extensively studied.
It should be clear that different and efficient solutions for the \IDP{} problem may highlight novel approaches for the above mentioned computational frameworks.


The approach of Pe'er et al.~\cite{pe2004incomplete} to the \IDP{} problem is graph theoretic: 
their algorithm relies on maintaining the connected components of a graph under a sequence of edge deletions.
The use of pre-existing dynamic connectivity data structures for this purpose is the bottleneck in the overall time complexity. 
A connectivity data structure is \emph{fully-dynamic} when both edge insertion and deletion are allowed, and \emph{decremental} when only edge deletion is considered.
A long line of results brought down the computational time required for updating the data structure after edge insertions and/or deletions, and for answering connectivity queries, to roughly logarithmic: the following table summarizes the results for both fully-dynamic and decremental connectivity 
on a graph consisting of $N$ nodes and $M$ edges.
For fully-dynamic connectivity we report the update time required for a single edge insertion or deletion, while for decremental connectivity we report the overall time required to eventually delete all of the edges.
All of the listed results, except for~\cite{henzinger1999constructing}, assume that edge deletions can be interspersed with connectivity queries.
The algorithm of Henzinger et al.~\cite{henzinger1999constructing}, in contrast, deletes edges in batches ($b_0$ is the number of batches that do not result in a new component) and connectivity queries can be only asked between one batch of deletions and another.

\[
\begin{array}{lll}
\toprule
\textbf{Fully-Dynamic}& \textbf{Update time} & \textbf{Query time} \\
\midrule
 \textnormal{Holm et al.~\cite{holm2001poly}}&  \cO(\log^2 N), \textnormal{amortized}   &  \cO(\log N / \log\log N)\\
 \textnormal{Gibb et al.~\cite{gibb2015dynamic}} &\cO(\log^4 N), \textnormal{worst case} & \cO(\log N/\log\log N)\,\,\textnormal{w.h.p.} \\
 \textnormal{Huang et al.~\cite{huang2017fully}} & \cO(\log N(\log\log N)^2), \textnormal{expected amortized} & \cO(\log N/\log\log\log N) \\
\midrule
\textbf{Decremental}& \textbf{Total update time} & \textbf{Query time}\\
\midrule
 \textnormal{Even et al.~\cite{EvenS81}} & \cO(MN) & \cO(1)\\
 \textnormal{Thorup~\cite{thorup1999decremental}} & \cO(M\log^{2}(N^{2}/M)+N\log^{3} N\log\log N), \textnormal{expected} & \cO(1)\\
 \textnormal{Henzinger et al.~\cite{henzinger1999constructing}} & \cO(N^2\log N+b_0\min\{N^2,M\log N\}) & \cO(1)\\
\bottomrule
\end{array}
\]

By plugging in an appropriate dynamic connectivity structure, the worst case running time of the approach of Pe'er et al.~\cite{pe2004incomplete},
given a matrix describing $n$ species and $m$ characters, becomes
deterministic $\cO(nm\log^{2}(n+m)$ (using fully dynamic connectivity structure of Holm et al.~\cite{holm2001poly}),
expected $\cO(nm\log((n+m)^{2}/nm)+(n+m)\log^{3}(n+m)\log\log(n+m))$ (using decremental connectivity structure of Thorup~\cite{thorup1999decremental}),
expected $\cO(nm\log(n+m)(\log\log(n+m))^{2})$ (using fully dynamic connectivity structure of Huang et al.~\cite{huang2017fully}), 
or deterministic $\cO((n+m)^{2}\log(n+m))$ (using decremental structure of Henzinger et al.~\cite{henzinger1999constructing}).
This should be compared with a lower bound of $\Omega(nm)$, following from the work of Gusfield on directed binary perfect phylogeny~\cite{gusfield1991efficient} (under the natural assumption that the input is given as a matrix).
For $n=m$, the second algorithm achieves this lower bound at the expense of randomisation (and being very complicated),
while for the general case the asymptotically fastest solution is still at least one log factor away from the lower bound.

Inspecting the algorithm of Pe'er et al.~\cite{pe2004incomplete}, we see that it operates on bipartite graphs
and only needs to deactivate nodes on one of the sides. It seems plausible that some of the known dynamic connectivity
structures are actually asymptotically more efficient on such instances. However, all of them are very complex
(with the result of Holm et al.~\cite{holm2001poly} being the simplest, but definitely not simple), and this is not clear.
Furthermore, recently Fern{\'a}ndez-Baca and Liu~\cite{fernandez2019tree} performed an experimental study of the algorithm of Pe'er et al. for \IDP~\cite{pe2004incomplete} with the aim of assessing the impact of the underlying dynamic graph connectivity data structure on their solution. 
Specifically, they tested the use of the data structure of Holm et al.~\cite{holm2001poly} against a simplified version of the same method, and showed that, in this context, simple data structures perform better than more sophisticated ones with better asymptotic bounds.

\paragraph{Our results and techniques. } We are motivated to look for simple, ad-hoc methods that make use of the properties of the
decremental connectivity as used in \IDP{}.
In this case, the graph is bipartite, and the required updates are vertex deletions from just one of the two sides.
We thus start by describing a simple data structure that dynamically maintains the connected components of a bipartite graph
with $N$ nodes on each side, whilst vertices are removed from one side of the graph.
The starting point for our solution is an application of a version of the sparsification technique of Eppstein et al.~\cite{eppstein1997sparsification}:
we define a hierarchical decomposition of the graph, and maintain a forest representing the connected components
of each subgraph in this decomposition. Recall that the original description of this technique focused on inserting and deleting edges,
while we are interested in deleting nodes (and only from one side of the graph). Therefore, the decomposition needs
to be appropriately tweaked for this particular use case.
This allows us to obtain an extremely simple data structure with $\cO(N^{2}\log N)$ total update time, which we show to imply
an $\cO(nm\log n)$ algorithm for \IDP{}.

The main technical part of our paper refines this solution to shave the logarithmic factor and thus obtain an asymptotically
optimal algorithm. We stress that while Eppstein et al.~\cite{eppstein1997sparsification} did manage to avoid paying
any extra log factors by applying a more complex decomposition of the graph than a complete binary tree (used in the
conference version of their paper), this does not seem to translate to our setting, as we operate on the nodes
instead of the edges.
The high-level idea is to amortize the time spent
on updating the forest representing the components of every subgraph with the progress in disconnecting its nodes, and re-use the results from
the subgraph on the previous level of the decomposition to update the subgraph on the next level.
As a consequence, the \IDP{} problem can be solved in time linear in the input size: 
\begin{restatable}{theorem}{linear}\label{the:nm}
    Given an incomplete matrix $\cA_{n\times m}$, the \IDP{} problem can be solved in time $\cO(nm)$.
\end{restatable}
Under the natural assumption that the input is given as a matrix, this is asymptotically optimal~\cite{gusfield1991efficient}.

\paragraph{Paper organization. }
In Section~\ref{sec:prel} we provide a description of the algorithm of Pe'er et al.~\cite{pe2004incomplete} and a series of preliminary observations.
In Section~\ref{sec:nmlogn} we show a simple and self-contained dynamic connectivity data structure that implies  an $\cO(nm\log n)$ time solution for
the \IDP{} problem for an incomplete matrix $\cA_{n\times m}$.
Finally, in Section~\ref{sec:nm} we present the main result of this paper and describe a dynamic connectivity data structure that implies a linear-time algorithm for \IDP{}.


\section{Preliminaries}\label{sec:prel}
\paragraph{Basic definitions.}Let $G=(V,E)$ be a graph.
The subgraph induced by $V' \subseteq V$ is the graph $G_{V'}=(V',E\cap(V'\times V'))$. 
We say that a forest $F=(V,E')$ represents the connected components of $G=(V,E)$ when the
connected components of $F$ and $G$ are the same (note that we do not require that $E'\subseteq E$).
Throughout the paper, we will use the term node for trees, and vertex for other graphs.
We denote by $S = \{s_1,\ldots,s_n\}$ the set of species and
by $C = \{c_1,\ldots,c_m\}$ the set of characters.
A matrix of character states $\cA_{n\times m}=[a_{ij}]_{n\times m}$, where each entry is a state from $\{0,1,?\}$ 
and the rows correspond to the species, is said to be \emph{incomplete}. The state
$a_{ij}$ of a character $j$ for a species $i$ is one, zero or $?$ depending on whether character $j$ is present, absent or unknown for species $i$. 
A completion $\cB_{n\times m}$ of 
such $\cA_{n\times m}$ is obtained by replacing the $?$ entries of $\cA_{n\times m}$ with either $0$ or $1$: formally, $\cB_{n\times m}$ is a binary matrix with entries $b_{ij}=a_{ij}$ for each $i,j$ such that $a_{ij}\neq \,\,?$.

A phylogenetic rooted tree $\cT$ for a binary matrix $\cB_{n\times m}$ has the $n$ species of $S$ at the leaves, and there is a surjection from
the set of characters $C$ to the internal nodes of $\cT$ such that,
if a character $c_j$ is associated with a 
node $x$, then $s_i$ belongs to the leaf set of the subtree rooted at $x$ if and only if $b_{ij} = 1$. In other words, all and only the species in a subtree associated with a character $c_j$ have the character $c_j$. We say that an incomplete matrix admits a phylogenetic tree if there exists a completion of the matrix that has such a tree. 
The Incomplete Directed Perfect Phylogeny problem  (\IDP{} for short), introduced by Pe'er et al. in~\cite{pe2004incomplete}, asks, given an incomplete matrix $\cA$, to find a phylogenetic tree for $\cA$, or determine that no such tree exists.

For a character $c_j$, the $1$-set (resp. $0$-set and $?$-set) of $c_j$ in an incomplete matrix $\cA$ is the set of species $\{s_i | a_{ij}=1\}$ (resp. $a_{ij}=0$ and $a_{ij}=\,\,?$). For a subset $S' \subseteq S$ of species, 
a character $c$ is $S'$-semiuniversal
in $\cA$ if its 0-set does not intersect $S'$, that is, if $\cA[s,c]\neq 0$ for all $s\in S'$. 
It is convenient to represent the character state matrix as a graph:
the vertices are $V=S\cup C$ and the edges are $S\times C$,  partitioned into $E_1\cup E_?\cup E_0$, with $E_x = \{(s_i, c_j ) | a_{ij} = x\}$ for $x\in\{0,1,?\}$. 
The edges of $E_1,E_{?},E_{0}$ are called \emph{solid}, \emph{optional}, and \emph{forbidden}, respectively.
We denote by $G(\cA)=(S \cup C, E_1)$ the bipartite graph consisting only of the solid edges.

\paragraph{Previous solutions.}The existence of a phylogenetic tree for $\cA$ is linked with the existence, in its graph representation, of a subset of edges with certain properties.
Specifically, Pe'er et al. show that finding a subset $D\subseteq (E_1\cup E_?)$ such that $E_1\subseteq D$
and $(S \cup C,D)$ is $\Sigma$-\emph{free} (where a $\Sigma$ is a path consisting of four edges
induced  by three vertices from $S$ and two vertices from $C$), or determining that such $D$ does not exist, is equivalent to solving the \IDP{} problem for $\cA$.


Pe'er et al. proposed two algorithms for solving the \IDP{} problem, both working on the graph representation of $\cA$ and relying on some graph dynamic connectivity data structure, the main difference between the two being the data structure they use. 
For ease of presentation, in what follows we will only consider the algorithm they refer to as Alg\_A. 
The algorithm relies on the following key properties: if an incomplete matrix $\cA$ admits a phylogenetic tree,
and $c$ is a $S$-semiuniversal character (meaning that there are no $0$s in its column),
then the incomplete matrix obtained by setting to $1$ all of the entries of column $c$ still admits a phylogenetic tree.
Moreover, given a partition $(K_1,\ldots,K_r)$ of $S\cup C$ where each $K_i$ is a connected component of $G(\cA)$, the incomplete matrix obtained by setting to $0$ all entries corresponding to the edges between $K_i$ and $K_j$, for $i\neq j$, still admits a phylogenetic tree.
Then, there is no interaction between the species and characters belonging to different connected components, and
the whole reasoning can be repeated on each such component separately.

We denote by $S(K)$ and $C(K)$ the set of species and characters, respectively, of a connected component $K$ of $G(\cA)$;
$\cA|_K$ denotes the submatrix of $\cA$ corresponding to the species and characters in $K$.
\emph{Deactivating} a character $c$ in $G(\cA)$ consists in deleting $c$ together with all its incident edges.
At a high level, Alg\_A works as follows. 
At each step, for each connected component $K_i$ of $G(\cA)$, it computes the $S(K_i)$-semiuniversal characters. 
If, for some $K_i$, no $S(K_i)$-semiuniversal character exists, it can be proven that, for any $D\subseteq  (E_{1}\cup E_{?})$
such that $E_{1}\subset D$, the graph $(S\cup C,D)$ is not $\Sigma$-free,
therefore the process halts and reports that $\cA$ does not admit a phylogenetic tree.
Otherwise, it sets to $1$ all of the entries of $\cA|_{K_i}$ corresponding to the $S(K_i)$-semiuniversal characters, and sets to $0$ the entries of $\cA$ between vertices that lay in different connected components. It then deactivates all of the $S(K_i)$-semiuniversal characters and updates the connected components of $G(\cA)$ using some dynamic connectivity data structure.

Algorithm~\ref{alg:peer} summarizes the high-level structure of Alg\_A:
for the sake of clarity, we only included the steps that compute the information needed for determining whether $\cA$ has a phylogenetic tree, and we left out the operations that actually reconstruct the tree. A complete pseudocode and a proof of correctness of the algorithms can be found in~\cite{pe2004incomplete}.

\begin{algorithm}
\setlength{\interspacetitleruled}{0pt}%
\setlength{\algotitleheightrule}{0pt}%
\While{\textnormal{there is at least one character in} $G(\cA)$}{
    Find the connected components of $G(\cA)$\;\label{line:find}
    \For{\textnormal{\textbf{each} connected component $K_i$ of $G(\cA)$ with at least one character}}{
        Compute the set $U$ of all characters in $K_i$ which are $S(K_i)$-semiuniversal in $\cA$\;\label{line:semi}
        \lIf{$U=\emptyset$}{\Return{\textnormal{\texttt{FALSE}}}}
        Deactivate every $c\in U$\;\label{line:deactivate}
    }
    \Return \texttt{TRUE}
}
\caption{The high-level structure of Alg\_A~\cite{pe2004incomplete}.}
\label{alg:peer}
\end{algorithm}
 
\paragraph{Preliminary results.}Our goal is to improve Alg\_A by optimizing its bottleneck, that is maintaining
the connected components of $G(\cA)$. We will represent the connected components
of a bipartite graph $G$ using the following lemma, and call the resulting representation a \emph{list-representation} of $G$.

\begin{lemma}\label{lem:cc}
The connected components of a bipartite graph $G=(S\cup C,E)$ can be represented in $\cO(|S|+|C|)$ space
so that, given a vertex, we can access its component, including the size and a pointer to the list of species and characters
inside, in constant time, and move a vertex to another component (or remove it from the graph) also in constant time.
\end{lemma}

\begin{proof}
Each component of $G$ is represented by a doubly-linked list of its vertices
(more precisely, a list of species and a list of characters), and also stores the size of the list.
An array of length $n+m$, indexed by the vertices of $G$, stores a pointer from each vertex to its component and
another pointer from each vertex to its position in the list of that component.
The components are, in turn, organised in a doubly-linked list.
Such representation takes space linear in the number of vertices
and allows us to access all the required information in constant time.
Further, removing or moving a vertex to another component takes constant time.
\end{proof}

Given a list-representation of $G$, we represent its connected components with
another graph $F=(V,E')$ consisting of \emph{rooted stars}~\cite{ShiloachV82} as follows.
For each component $K$, we define the central vertex
$v\in K$ to be the first vertex on the list of $K$. Then, we add an edge $(u,v)$ to $E'$, for any $u\in K$ with $u\neq v$.
This construction can be implemented in $\cO(|V|)$ time. Observe that we
only guarantee that the connected components of $G$ and $F$ are the same,
but $E'$ is not required to only consist of the edges of $G$.
We can use the list-representation of $G$ to simulate access to
the adjacency lists of $F$ without constructing it explicitly, as stated by the following lemma.

\begin{lemma}\label{lem:simulation}
Given a bipartite graph $G=(S\cup C,E)$ and a list-representation of $G$, the access to the adjacency lists of a star forest $F$ representing the connected components of $G$ can be simulated in constant time without constructing $F$ explicitly.
\end{lemma}
\begin{proof}
To access the adjacency list of a vertex $v$ we first look up its component $K$ and retrieve the first vertex $u$ on the
list of $K$. By Lemma~\ref{lem:cc}, this operation requires constant time. If $u=v$, then the adjacency list of $v$ is the list of vertices of $K$ stored in the list-representation of $G$.
Otherwise, the adjacency list of $v$ consists only of a single vertex $u$. 
\end{proof}

We are interested in solving the following special case of decremental connectivity:

\Dynproblem{$(N_{\ell},N_{r})$-\DC{}}
{a bipartite graph $G=(S\cup C,E)$ with $N_{\ell}=|S|$ and $N_{r}=|C|$.}
{deactivate a character $c\in C$.}
{return the connected components of the subgraph induced by $S$ and the remaining characters.}

When analysing the complexity of $(N_{\ell},N_{r})$-\DC{}, we allow preprocessing the input graph $G$
in $\cO(N_{\ell}N_{r})$ time, and assume that all characters are eventually deactivated when analysing the total
update time. We can of course deactivate multiple characters at once by deactivating them one-by-one.

The overall time complexity of Algorithm~\ref{alg:peer} depends on the complexity of $(N_{\ell},N_{r})$-\DC{} as follows.

\begin{lemma}\label{lem:bottleneck}
Consider an $n\times m$ incomplete matrix $\cA$.
If the $(n,m)$-\DC{} problem can be solved in $f(n,m)$ total update time
and $g(n,m)$ query time, then
the \IDP{} problem can be solved for $\cA$ in time $\cO(nm+f(n,m)+\min\{n,m\}\cdot g(n,m))$.
\end{lemma}

\begin{proof}
There are three nontrivial steps in every iteration of the while loop:
finding the connected components in line~\ref{line:find}, computing the semiuniversal characters
of every connected component in line~\ref{line:semi}, and finally deactivating characters in line~\ref{line:deactivate}.
Every character is deactivated at most once, so the overall complexity of all deactivations is
$\cO(f(n,m))$.
We claim that in every iteration of the while loop, except possibly for the very last, (1) at least
one character is deactivated, and (2) there exist two species that cease to belong to
the same connected component. (1) is immediate, as otherwise we have a connected component
$K_{i}$ with no $S(K_{i})$-semiuniversal characters and the algorithm terminates.
To prove (2), assume otherwise, then we have a connected component $K_{i}$ such that
$S(K_{i})$ does not change after deactivating all $S(K_{i})$-semiuniversal characters. But then
in the next iteration the set of $S(K_{i})$-semiuniversal characters is empty and the algorithm
terminates. (1) and (2) together imply that the number of iterations is bounded by $\min\{n,m\}$.
The overall complexity of finding the connected components is thus $\cO(\min\{n,m\}\cdot g(n,m))$.

It remains to bound the overall complexity of computing the semiuniversal characters by $\cO(nm)$. 
This has been implicitly done in~\cite[proof of Theorem 12]{pe2004incomplete}, but we provide
a full explanation for completeness.
For every character $c\in C$, we maintain the count of solid and optional edges connecting
$c$ (in the graph representation of $\cA$) with the species that belong to its same connected component (of $G(\cA)$).
Assuming that we can indeed maintain these counts, in every iteration all the semiuniversal
characters can be generated in $\cO(m)$ time, so in $\cO(\min\{n,m\}\cdot m)=\cO(nm)$ overall time.

To update the counts, consider a connected component $K$ that, after deactivating some characters,
is split into possibly multiple smaller components $K_{1},K_{2},\ldots,K_{k}$.
Note that we can indeed gather such information in $\cO(n+m)$ time, assuming access to
a representation of the connected components before and after the deactivation.
We assume that the connected components are maintained with the list-representation described in Lemma~\ref{lem:cc},
and therefore we can access a list of the vertices in every $K_{i}$. Then, we consider every
pair $i,j\in \{1,2,\ldots,k\}$ such that $i\neq j$, $C(K_{i})\neq \emptyset$ and
$S(K_{j})\neq \emptyset$. We iterate over every $c\in K_{i}$ and $s\in K_{j}$,
and if $(s,c)$ is an edge in the graph of $\cA$ (observe that it cannot be a solid edge, as $K_{i}$ and $K_{j}$ are distinct
connected components) we decrease the count of $c$.
By first preparing lists of components $K_{i}$ such that $C(K_{i})\neq \emptyset$ and $S(K_{i})\neq \emptyset$,
this can be implemented in time bounded by the number of considered possible edges $(s,c)$,
and every such possible edge is considered at most once during the whole execution.
Therefore, the overall complexity of maintaining the counts is $\cO(nm)$.
Additionally, we need $\cO(nm)$ time to initialise the $(n,m)$-\DC{} structure.
\end{proof}

Before we proceed to design an efficient solution for the $(N_{\ell},N_{r})$-\DC{} problem, we first show that it is in fact
enough to consider the $(N,N)$-\DC{} problem.

\begin{lemma}\label{lem:red}
Assume that the $(N,N)$-\DC{} problem can be solved in $f(N)$ total update time and $g(N)$ query time.
Then, for any $N' \geq N$, both the $(N,N')$-\DC{} problem and the $(N',N)$-\DC{} problem
can be solved in $\cO(N'/N\cdot f(N))$ total update time and 
$\cO(N'/N\cdot g(N))$ query time.
\end{lemma}

\begin{proof}
We first consider the $(N,N')$-\DC{} problem.
We create $\lceil N'/N\rceil $ instances of $(N,N)$-\DC{} by partitioning $C$ into groups of $N$ vertices (except for
the last group that might be smaller). In each instance we have the same set of species $S$. Deactivating
a character $c\in C$ is implemented by deactivating it in the corresponding instance of $(N,N)$-\DC{}.
Overall, this takes $\cO(N'/N\cdot f(N))$ time.
Upon a query, we query all the instances in $\cO(N'/N\cdot g(N))$ time. The output of each
instance can be converted to a star forest representing the connected components in $\cO(N)$ time.
We take the union of all these forests to obtain an auxiliary graph on at most $\lceil N'/N\rceil \cdot (N-1) =\cO(N')$ edges,
and find its connected components in $\cO(N')$ time.
Assuming that $f(N) \geq N$, this takes $\cO(N'/N\cdot f(N))$ overall time and gives us the connected
components of the whole graph.

Now we consider the $(N',N)$-\DC{} problem.
We create $\lceil N'/N\rceil $ instances of $(N,N)$-\DC{} by partitioning $S$ into groups of $N$ vertices,
and in each instance we have the same set of characters $C$. Thus, deactivating a character $c\in C$
is implemented by deactivating it in every instance. Overall, this takes $\cO(N'/N\cdot f(N))$ time.
A query is implemented exactly as above by querying all the instances and combining the results
in $\cO(N'/N\cdot f(N))$ time.
\end{proof}

\section{$(N,N)$-\DC{} in $\cO(N^{2}\log N)$ Total Update Time and $\cO(N)$ per Query}
\label{sec:nmlogn}

Our solution for the $(N,N)$-\DC{} problem is based on a hierarchical decomposition of $G$ into multiple smaller
subgraphs as in the sparsification technique of Eppstein et al.~\cite{eppstein1997sparsification} (as
mentioned in the introduction, appropriately tweaked for our use case).
The decomposition is represented by a complete binary tree $\DT(G)$ of depth $\log N$. 
We identify the leaves of $\DT(G)$ with the characters $C$. 
Each node $v$ corresponds to the set of characters $C_{v}$ identified with the leaves in the
subtree of $v$, and is responsible for the subgraph $G_{v}$ of $G$ induced by $C_{v}$ and the whole set of species $S$.
Thus, the root is responsible for the whole $G$, see Figure~\ref{fig:dtree}.
Each node $v$ explicitly maintains a list-representation of the connected components of $G_{v}$, denoted $\comp(v)$.
We stress that, while $\comp(v)$ is explicitly maintained, we do not explicitly store $G_{v}$ at
every node $v$.
The initial preprocessing required to construct $\DT(G)$ together with $\comp(v)$ for every
node $v$, given $G$, takes $\cO(B^{2})$ time by the following argument. First, we construct $\comp(v)$
for every leaf $c$. This can be done in $\cO(B)$ time per leaf by simply iterating the neighbours
of $c$ in $G$. Second, we proceed bottom-up and compute $\comp(v)$ for every inner node
$v$ in $\cO(B)$ time using the following lemma.

\begin{lemma}
\label{lem:union}
Let $v$ be an inner node of $\DT(G)$, and $v_{\ell},v_{r}$ be its children.
Given $\comp(v_{\ell})$ and $\comp(v_{r})$ we can compute $\comp(v)$ in $\cO(B)$ time.
\end{lemma}

\begin{proof}
We construct star forests representing the connected components of $\comp(v_{\ell})$ and $\comp(v_{r})$ in $\cO(B)$
time and take their union. Then we find the connected components of this union in $\cO(B)$ time
and save them as $\comp(v)$.
\end{proof}

\begin{figure}[t]
    \centering
   \includegraphics[width=.7\linewidth]{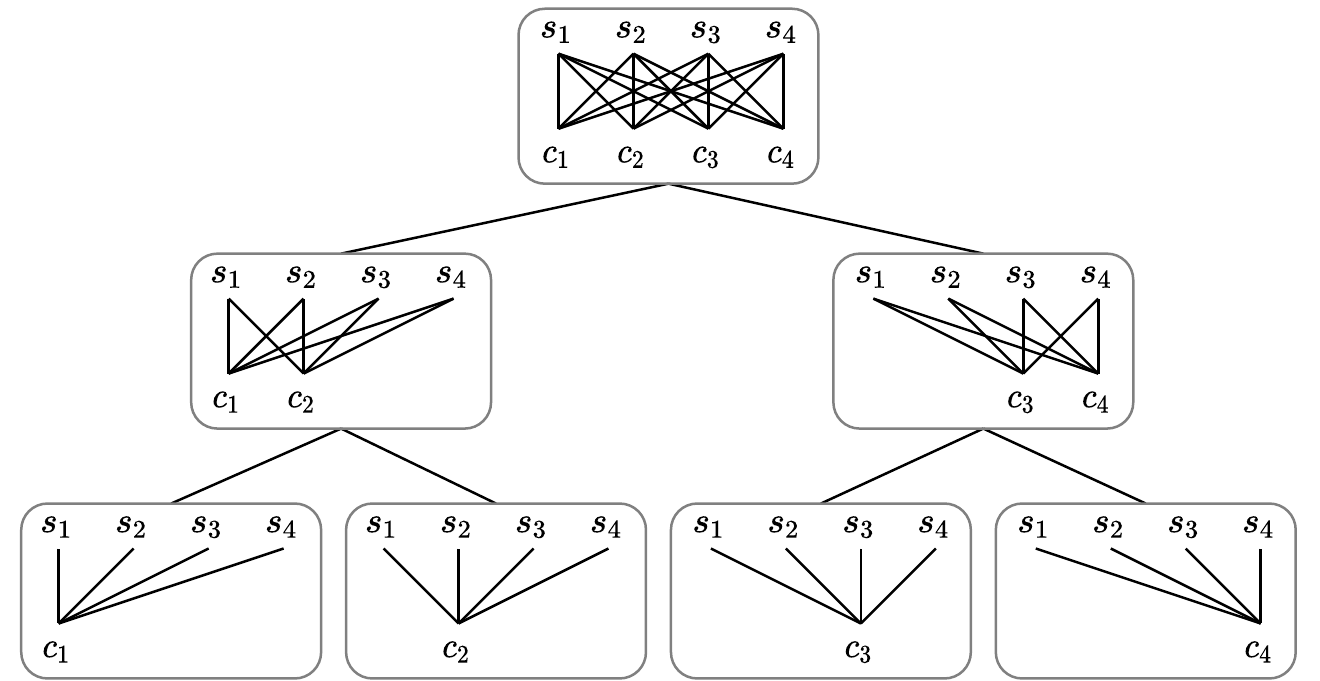}
    \caption{The decomposition tree of $K_{4,4}$.}
    \label{fig:dtree}
\end{figure}

We proceed to explain how to solve the $(N,N)$-\DC{} problem in $\cO(N\log N)$ time per update and
$\cO(N)$ time per query. The query simply returns $\comp(r)$, where $r$ is the root of $\DT(G)$.
The update is implemented as follows. Deactivating a character $c$ possibly affects $\comp(v)$ for
all ancestors $v$ of the leaf corresponding to $c$. In particular, $\comp(c)$ becomes a collection of
isolated nodes and can be recomputed in $\cO(1+|S|)=\cO(N)$ time.
We iterate over all proper ancestors $v$, starting from the parent of $c$. For each such
$v$, let $v_{\ell}$ and $v_{r}$ denote its left and right child, respectively. We can assume that
$\comp(v_{\ell})$ and $\comp(v_{r})$ have been already correctly updated. We compute $\comp(v)$
from $\comp(v_{\ell})$ and $\comp(v_{r})$ by applying Lemma~\ref{lem:union} in $\cO(N)$ time.
When summed over all the ancestors, the update time becomes $\cO(N\log N)$, so $\cO(N^{2}\log N)$
over all deactivations.

By Lemmas~\ref{lem:bottleneck} and~\ref{lem:red}, this implies that, given an incomplete matrix $\cA_{n\times m}$, the \IDP{} problem can be solved in time $\cO(nm\log \min\{n,m\})$ without using any dynamic connectivity data structure as a blackbox.

\section{$(N,N)$-\DC{} in $\cO(N^{2})$ Total Update Time and $\cO(N)$ Time per Query}
\label{sec:nm}

Our faster solution is also based on a hierarchical decomposition $\DT(G)$ of $G$. As before, every
node $v$ stores $\comp(v)$, so a query simply returns $\comp(r)$. The difference is in implementing
an update. We observe that, if for some ancestor $v$ of the leaf corresponding to $c$, the only change
to $\comp(v)$ is removing $c$ from its connected component, then this also holds for all of the
subsequent ancestors and they can be updated in constant time each.
This suggests that we should try to amortise the cost of an update with the progress in splitting
$\comp(v)$ into smaller components.

We will need to compare the situation before and after the update, and so introduce the following
notation. A node $v$ of $\DT(G)$ is responsible for the subgraph $G_{v}$ before the update
and for the subgraph $G'_{v}$ after the update; $\comp(v)$ and $\comp'(v)$ denote
the connected components of $G_{v}$ and $G'_{v}$, respectively. The crucial observation is
that $\comp'(v)$ is obtained from $\comp(v)$ by removing $c$ from its connected
component and, possibly, splitting this connected component into multiple smaller ones, while leaving the others intact.

Deactivating a character $c$ begins with updating naively $\comp(c)$ in $\cO(N)$ time. Then we iterate
over the ancestors of $c$ in $\DT(G)$. Let $v_{i+1}$ be the currently considered ancestor, $v_{i}$
the ancestor considered in the previous iteration, and $u_{i}$ be the other child of $v_{i+1}$ (sibling of $v_{i}$).
Let the component of $G_{v_{i}}$ containing $c$ be $K$. As observed above, the components of
$G'_{v_{i}}$ are the same as the components of $G_{v_{i}}$, except that $K$ is replaced by possibly
multiple components $K_{1},K_{2},\ldots,K_{k}$, where $\bigcup_{j=1}^{k}K_{j} = K\setminus\{c\}$.
If $k=1$ then we trivially remove $c$ from its connected component in every $G_{v_{j}}$, for $j=i+1,i+2,\ldots$
and terminate the update, so we can assume that $k\geq 2$. We further assume that, after having updated
the components of $G_{v_{i}}$, we obtained a list of pointers to $K_{1},K_{2},\ldots,K_{k}$.
Let $L$ be the connected component of $c$ in $G_{v_{i+1}}$, with $K\subseteq L$ because the subgraphs
are monotone with respect to inclusion on any leaf-to-root path. Now the goal is to transform $G_{v_{i+1}}$ into $G'_{v_{i+1}}$, to
update its components (using $\comp'(v_{i})$ and $\comp(u_{i})$), and additionally to obtain a list of pointers
to the components obtained by splitting $L$. See Figure~\ref{fig:components} for an illustration. 

\begin{figure}[h]
\begin{center}
\includegraphics[width=0.6\textwidth]{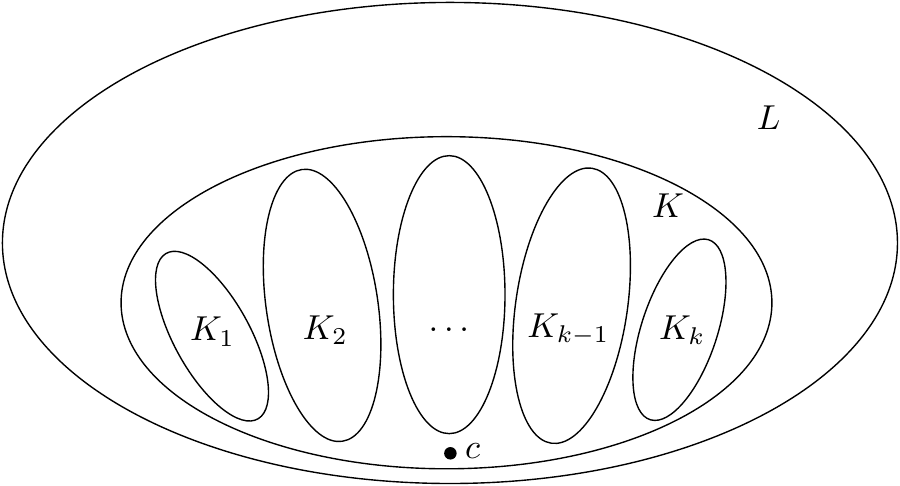}
\end{center}
\caption{After having removed $c$ from $K$ to obtain $K_{1},K_{2},\ldots,K_{k}$, we want
to remove $c$ from $L$.}
\label{fig:components}
\end{figure}

We start by initialising $G'_{v_{i+1}}$ to be $G_{v_{i+1}}$, and by removing $c$ from $L$.
As in the proof of Lemma~\ref{lem:union}, we will work with an auxiliary graph consisting of the union
of two star forests representing the connected components of $G'_{v_{i}}$ and $G_{u_{i}}$, respectively.
However, instead of explicitly constructing these forests, we simulate access to the adjacency lists of
every vertex in both forests using $\comp'(v_{i})$ and $\comp(u_{i})$, as explained in the proof of Lemma~\ref{lem:simulation}.
In turn, this allows us to simulate access to the adjacency list of every vertex in the auxiliary graph.
See Figure~\ref{fig:update} for an example of the auxiliary graph.

By renaming the components we can assume that $|K_{1}| \geq |K_{2}|,|
K_{3}|,\ldots,|K_{k}|$. 
We will visit the vertices of $L$ in order to determine the new connected components after the removal of $c$: when doing so, 
we will use different colours to represent vertices whose new connected component contains $K_1$ (red), vertices whose new component is different from the one of $K_1$ (black) and vertices whose new component is still unknown (white).
Initially, the vertices of $K_{1}$ are red and all of the other vertices of the auxiliary graph are white. This initialisation
is done implicitly, meaning that we will assume that all the vertices of $K_1$ are red and the rest are white without explicitly assigning the colours,
and whenever retrieving the colour of a node $u$ we first check if $u\in K_{1}$, and if so assume that it is red.
This allows us to implement the initialisation in constant time instead of $\cO(N)$ time. 
We will perform the visit of $L$ by running the following search procedure from an arbitrarily chosen vertex of each $K_{j}$, for $j=2,3,\ldots,k$.

The search procedure run from a vertex $x$ first checks if $x$ is white,
and immediately terminates otherwise. Then, it starts visiting the vertices
of the connected component of $x$ in the auxiliary graph: at any moment, each vertex in such component is either white or red.
As soon as the search encounters a red vertex, it is terminated and all the vertices visited in the current
invocation are explicitly coloured red. Otherwise, the procedure has identified a new connected component
$K'$ of $G'_{v_{i+1}}$. The vertices of $K'$ are removed from $L$, all vertices of $K'$
are coloured black in the auxiliary graph, and a new component $K'$ of $G'_{v_{i+1}}$ is created in $\cO(|K'|)$ time.
Inspect Figure~\ref{fig:update} for an example.

\begin{figure}[b]
    \centering
   \includegraphics[width=.6\linewidth]{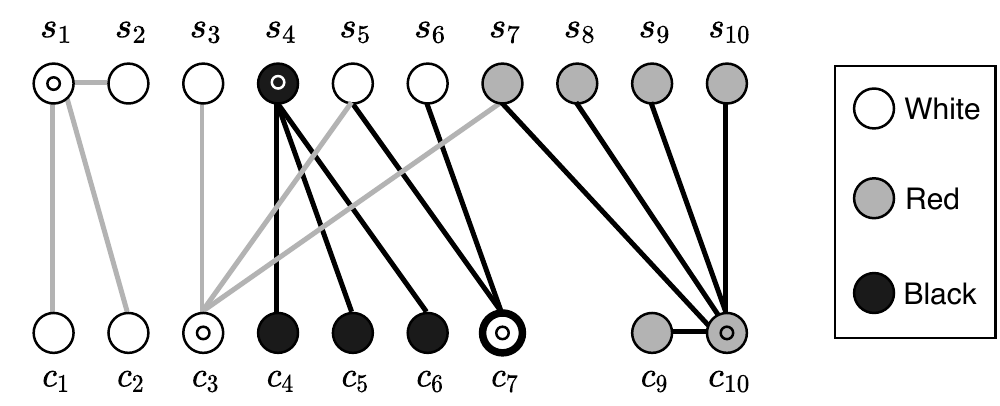}
    \caption{The auxiliary graph implicitly constructed for a node $v_{i+1}$ after deactivating $c_8$. Black edges are used for the star forest of $v_i$, grey edges for the star forest of $u_i$; an inner circle identifies the central vertices. $K_1$ is the rightmost component; $c_7$ is the next vertex to be considered, and it will eventually become red.}
    \label{fig:update}
\end{figure}

\begin{lemma}
\label{lem:search}
The total time spent on all calls to the search procedure in the current iteration is $\cO(|L|-|K_{1}|)$.
\end{lemma}

\begin{proof}
All vertices visited in the current iteration belong to $L$. The search is terminated as soon as
we encounter a red vertex, and all vertices of $K_{1}$ are red from the beginning. Therefore, each run of
the search procedure encounters at most one vertex of $K_{1}$, and we can account for traversing
the edge leading to this vertex separately paying $\cO(k-1)=\cO(|L|-|K_{1}|)$ overall. It remains to
bound the number of all other traversed edges. This is enough to bound the overall time of the
traversal, because every edge is traversed at most twice, and the number of visited isolated vertices
is at most $k-1=\cO(|L|-|K_{1}|)$.

For any other edge $e=\{u,v\}$, we have $u,v\in L$ but $u,v\notin K_{1}$. These edges can
be partitioned into two forests by considering whether they originate from $\comp'(v_{i})$
or $\comp(u_{i})$. Consequently, we must analyse the total number of edges in a union
of two forests spanning $L\setminus K_{1}$. But this is of course $\cO(|L|-|K_{1}|)$,
proving the lemma.
\end{proof}

We now need to analyse the sum of $|L|-|K_{1}|$ over all the iterations. Because $\bigcup_{j=1}^{k}K_{j}\subseteq L$,
we can split this expression into two parts:
\begin{enumerate}
\item $L\setminus \bigcup_{j=1}^{k}K_{j}$, 
\item $\sum_{j=2}^{k}|K_{j}|$.
\end{enumerate}
Because the sets $L\setminus \bigcup_{j=1}^{k}K_{j}$ considered in different iterations
are disjoint, the first parts sum up to $\cO(n)$. It remains to
bound the sum of the second parts. This will be done by the following argument.
Consider an arbitrary $G_{v}$ corresponding to a subgraph induced by all the species and a subset of $2^{d}$ characters.
Whenever its connected component $K$ is split into smaller connected components
$K_{1},K_{2},\ldots,K_{k}$ after deactivating a character $c$ in the subtree of $v$, the second part $\sum_{j=2}^{k}|K_{j}|$
is distributed among the vertices of $\bigcup_{j=2}^{k}K_{j}$. That is, each node of $\bigcup_{j=2}^{k}K_{j}$ pays 1.
Observe that the size of the connected component containing such a node decreases by a factor of at least 2,
because $|K_{2}|,|K_{3}|,\ldots,|K_{k}| \leq |K|/2$.
To bound the sum of second parts, we analyse the total cost paid by all the nodes of $G_{v}$
due to deactivating the characters in the subtree of $v$ (recall that in the end all such characters
are deactivated).

\begin{lemma}
The total cost paid by the nodes of $G_{v}$, over all $2^{d}$ deactivations affecting $v$, is $\cO(N\cdot d)$.
\end{lemma}

\begin{proof}
We claim that in the whole process there can be at most $2^{t+1}$ deactivations incurring a cost from $[N/2^{t+1},N/2^{t})$.
Assume otherwise, then there exists a vertex $x$ charged twice by such deactivations. As a result of
the first deactivation, the size of the connected component containing $x$ drops from less than $N/2^{t}$ to below $N/2^{t+1}$.
Consequently, during the next deactivation that charges $x$ the cost must be smaller than $N/2^{t+1}$, a contradiction.
As we have $2^{d}$ deactivation overall, the total cost can be at most:
\[ \sum_{t=0}^{d} 2^{t+1}\cdot N/2^{t} = \cO(N\cdot d) \]
as claimed.
\end{proof}

To complete the analysis, we observe that there are $N/2^{d}$ nodes of $\DT(G)$ such that we have
$2^{d}$ deactivations affecting $v$. The sum of the second parts is thus:
\[ \sum_{d=0}^{\log n} N/2^{d} \cdot n\cdot d  < N^{2} \sum_{d=0}^{\infty} d/2^{d} = \cO(N^{2}).  \]
Overall, the total update time is hence $\cO(N^{2})$.
By Lemmas~\ref{lem:bottleneck} and~\ref{lem:red}, it implies the following:
\linear*

\bibliographystyle{plainurl}
\bibliography{references}

\end{document}